\documentclass[11pt,a4paper,reqno]{amsart}%
\usepackage{amsthm,amsmath,amsfonts,amssymb,amsxtra,appendix,bookmark,dsfont,latexsym}

\makeatletter
\usepackage{hyperref}
\usepackage{delarray,a4,color}
\usepackage{euscript}
\usepackage[latin1]{inputenc}

\usepackage{pdfsync}
\numberwithin{equation}{section}

\newtheorem{theorem}{Theorem}[section]

\newtheorem{corollary}[theorem]{Corollary}
\newtheorem{proposition}[theorem]{Proposition}

\theoremstyle{definition}

\theoremstyle{remark}

\usepackage{color}

\newcommand{\bdm}{\begin{displaymath}}
\newcommand{\edm}{\end{displaymath}}
\newcommand{\bdn}{\begin{eqnarray}}
\newcommand{\edn}{\end{eqnarray}}
\newcommand{\bay}{\begin{array}{c}}
\newcommand{\eay}{\end{array}}
\newcommand{\ben}{\begin{enumerate}}
\newcommand{\een}{\end{enumerate}}
\newcommand{\beq}{\begin{equation}}
\newcommand{\eeq}{\end{equation}}

\newcommand{\eps}{\varepsilon}

\newcommand{\dist}{\mathrm{dist}}

\newcommand{\norm}[1]{\left\lVert #1 \right\rVert}

\newcommand{\R}{\mathbb{R}}
\newcommand{\N}{\mathbb{N}}

\newcommand{\F}{\mathcal{F}}
\newcommand{\E}{\mathcal{E}}

\newcommand{\B}{\mathcal{B}}

\newcommand{\I}{\mathcal{I}}

\newcommand{\PP}{\mathcal{P}}

\newcommand{\xbf}{\mathbf{x}}

\newcommand{\one}{{\ensuremath {\mathds 1} }}

\newcommand{\half}{\frac{1}{2}}

\newcommand{\intR}{\int_{\R ^2}}

\newcommand{\rhoP}{\rho_{\Psi}}
\newcommand{\rhoPF}{\rho_{\Psi_F}}

\newcommand{\Barg}{\B}
\newcommand{\BargN}{\B ^{N}}

\newcommand{\cLau}{c _{\rm Lau}}
\newcommand{\intN}{\I_N}

\newcommand{\Ker}{\mathrm{Ker} (\intN)}

\newcommand{\ZN}{\mathcal{Z}_{N}}

\newcommand{\PsiLau}{\Psi_{\rm Lau}}

\newcommand{\rhoF}{\rho_F}

\newcommand{\VD}{\mathcal{V}_2 ^D}

\newcommand{\rhot}{\tilde{\rho}}

\title{Incompressibility estimates for the Laughlin phase, part II}

\author[N. Rougerie]{Nicolas ROUGERIE}
\address{Universit\'e Grenoble 1 \& CNRS ,  LPMMC (UMR 5493), B.P. 166, F-38 042 Grenoble, France}
\email{nicolas.rougerie@grenoble.cnrs.fr}

\author[J. Yngvason]{Jakob Yngvason}
\address{Fakult\"at f\"ur Physik, Universit{\"a}t Wien, Boltzmanngasse 5, 1090 Vienna, Austria \&
Erwin Schr{\"o}dinger Institute for Mathematical Physics, Boltzmanngasse 9, 1090 Vienna, Austria.}
\email{jakob.yngvason@univie.ac.at}

\date{November 3, 2014}

\begin{document} 
\maketitle

\begin{abstract}
We consider fractional quantum Hall states built on Laughlin's original $N$-body wave-functions, i.e.,  they are of the form holomorphic $\times$ gaussian and vanish when two particles come close, with a given polynomial rate. Such states appear naturally when looking for the ground state of 2D particles in strong magnetic fields, interacting via repulsive forces and subject to an external potential due to trapping and/or disorder. We prove that all functions in this class satisfy a universal local density upper bound, in a suitable weak sense. Such bounds are useful to investigate the response of fractional quantum Hall phases to variations of the external potential. Contrary to our previous results for a restricted class of wave-functions, the bound we obtain here is not optimal, but it does not require any additional assumptions on the wave-function, besides analyticity and symmetry of the pre-factor modifying the Laughlin function.     
\end{abstract}

\tableofcontents

\section{Introduction}\label{sec:intro}

\subsection{Background}\label{sec:motiv} The fractional quantum Hall effect (FQHE)~\cite{StoTsuGos-99,Girvin-04} is a signature of the transition to a strongly correlated quantum  fluid in the ground state of a 2D electron gas submitted to a strong perpendicular magnetic field.  The best known of these states is described by Laughlin's wave function~\cite{Laughlin-83}
\begin{equation}\label{eq:intro Laughlin}
\PsiLau (z_1,\ldots,z_N) = \cLau \prod_{1\leq i<j \leq N} (z_i-z_j) ^{\ell}  e ^{- \sum_{j=1} ^N |z_j| ^2 / 2}. 
\end{equation}
Here, $z_1,\ldots,z_N$ are the coordinates, identified with complex numbers, of $N$ quantum particles in 2D. Since electrons are fermions, the positive integer $\ell$ should be an odd number for symmetry reasons. If $\ell=1$ the function \eqref{eq:intro Laughlin} is simply a Slater determinant, so for interacting fermions $\ell\geq 3$ is the important case. We shall also consider the case where $\ell$ is even so the wave function describes bosonic particles; this case is in particular relevant for ultra-cold Bose gases with short range repulsive interactions  in rapidly rotating traps \cite{LewSei-09,RouSerYng-13b, RouSerYng-13}. 

The form of the  function \eqref{eq:intro Laughlin} is dictated by two requirements:
\begin{itemize}
\item The function should belong to the lowest Landau level of a magnetic Laplacian, whence the form ``holomorphic $\times$ gaussian''. We have chosen units so that the width of the gaussian is fixed, i.e. does not depend explicitly on the magnetic field strength or the rotation speed.
\item The particles want to avoid one another in order to reduce as much as possible the interaction energy for a given repulsive potential. The holomorphic factor in~\eqref{eq:intro Laughlin} is thus chosen to vanish along the diagonals $z_i=z_j$ with a chosen polynomial rate~$\ell\in \N$. 
\end{itemize}
The constant $\cLau$ normalizes the function in $L^2 (\R^{2N})$ to make it a state vector that is a promising candidate for approximating the ground state of  2D quantum particles in a strong magnetic field, with strong repulsive interaction between the particles. More generally, in an external potential, natural candidates live in the space 
\begin{equation}\label{eq:Ker}
\mathcal L_\ell^N:= \left\{  \Psi\in L^2(\R ^{2N}):\  \Psi (z_1,\ldots,z_N) = \PsiLau (z_1,\ldots,z_N) F (z_1,\ldots,z_N), \: F \in \BargN \right\}. 
\end{equation}
Here $\BargN$ denotes the $N$-body bosonic Bargmann space 
\begin{equation}\label{eq:BargN}
\BargN :=  \left\{ F \mbox{ holomorphic and symmetric } | \:  F(z_1,\ldots,z_N) e^{- \sum_{j=1} ^N |z_j| ^2 /2  } \in L ^2 (\R ^{2N})\right\}.
\end{equation} 
Note that the symmetry constraint on $F$ implies that $\mathcal{L}_\ell ^N$ is an appropriate function space for fermions (resp. bosons) if $\ell$ is odd (resp. even). We thus treat both type of particles in the same way and are not able to exploit the effect of statistics beyond that of fixing the parity of $\ell$.

The functions in $\mathcal L_\ell^N$ are annihilated by two-body contact interactions \cite{TruKiv-85,PapBer-01}, and for $\ell = 2$ this space is exactly the kernel $\Ker$ of the contact interaction Hamiltonian
$$ \intN := \sum_{1\leq i < j \leq N} \delta_{ij}$$
acting on the $N$-body bosonic lowest Landau level by the prescription 
$$ \delta_{i,j} \left(F(z_1,\ldots,z_N) e^{- \sum_{j=1} ^N |z_j| ^2 /2} \right)= F\left(z_1,\ldots,\frac{z_i+z_j}{2},\ldots,\frac{z_i+z_j}{2},\ldots,z_N\right) e^{- \sum_{j=1} ^N |z_j| ^2 /2}. $$
We refer to our previous papers~\cite{RouSerYng-13,RouSerYng-13b} and in particular~\cite{RouYng-14} for more background and references. In particular, our notion of  ``incompressibility estimates''  is 
explained in~\cite{RouYng-14} and we recall it briefly in the next two paragraphs.

The variational space $\mathcal L_\ell^N$ is of course huge and it is of importance in practice to be able to determine which state will emerge as a natural candidate for the ground state in a given situation. All states in $\mathcal L_\ell^N$ have by definition the same magnetic kinetic energy, and since they all vanish with (at least) the same rate $|z_i-z_j| ^\ell$ along the diagonals of configuration space it is natural to expect that they will also have roughly the same (small) interaction energy for a wide class of interaction potentials. The emergence of one particular state as a candidate for a ground state must thus be a consequence of effects so far neglected: in quantum Hall systems\footnote{And also in the rotating Bose gases suggested to mimic them, see~\cite{RouYng-14} and references therein.}, these come mostly from the energetic contribution of external fields, due to trapping/and or disorder in a given sample. In other words, it seems reasonable 
that the emergence of a particular state should follow from the minimization of some reduced energy functional  on  $\mathcal L_\ell^N$. 

To be able to tackle such a problem, it is useful to know what properties (normalized) wave functions of $\mathcal L_\ell^N$ have in common. One such property should be an incompressibility, or rigidity, bound: there is good reason to believe that a $\Psi \in \mathcal L_\ell^N$ with $\norm{\Psi}_{L^2} = 1$ satisfies
\begin{equation}\label{eq:incomp formel}
\rhoP \leq \frac{1}{\pi \ell N},
\end{equation}
at least in a suitable weak sense, where
\begin{equation}\label{eq:intro density}
\rhoP (z):= \int_{\R ^{2(N-1)}} |\Psi (z,z_2,\ldots,z_N)| ^2 dz_2 \ldots dz_N 
\end{equation}
is the one-body probability density of $\Psi$. Note that~\eqref{eq:incomp formel} is in fact optimal in the sense that it is saturated for certain  choices of $\Psi$, in particular the Laughlin functions themselves, in the large $N$ limit, see~\cite{RouSerYng-13b}. Note, however, that \eqref{eq:incomp formel} is not expected to hold point-wise for finite $N$ \cite{Ciftja-06}.

This paper is devoted to the proof of estimates in this spirit. It is a sequel to~\cite{RouYng-14} where we introduced a suitable weak formulation of~\eqref{eq:incomp formel} and proved that it is satisfied for $\Psi$ of the particular form 
$$ \Psi (z_1,\ldots,z_N) = \PsiLau (z_1,\ldots,z_N) F (z_1,\ldots,z_N), \:F \in \VD, \: \norm{\Psi}_{L^2} = 1,$$
where ($\deg$ denotes the degree of a polynomial) $F$ belongs to the set
\begin{multline}\label{eq:var set 2 D}
\VD = \Big\{ F \in \Barg ^N \, : \, \mbox{ there exist } (f_1,f_2) \in \Barg \times \Barg ^2, \deg (f_1) \leq D N, \deg(f_2) \leq D, 
\\  F (z_1,\ldots,z_N) = \prod_{j = 1} ^N f_1 (z_j)\prod_{1\leq i < j \leq N } f_2 (z_i,z_j) \Big\}.
\end{multline}
In ~\cite{RouYng-14} it is argued that the bound should hold more generally if the pre-factor $F$ is a symmetrized product of functions that depend only on a fixed number $n$ of variables with $n\ll N$.

The main result of the present paper is a (weak version of a) bound of the form~\eqref{eq:incomp formel} without {\it any}  a priori assumption on $F$ except analyticity and symmetry. We prove that, in a sense made precise in Theorem \ref{thm:main} below,
\begin{equation}\label{eq:incomp formel 2}
 \rhoP \leq \frac{4}{\pi \ell N} 
\end{equation}
for any $\Psi \in  \mathcal L_\ell^N$ with $\norm{\Psi}_{L^2} = 1$. Note that the constant on the right-hand side of \eqref{eq:incomp formel 2} is four times the one in ~\eqref{eq:incomp formel}. 
\medskip

In brief, we prove in the present paper  incompressibility bounds for {\it all\/} functions in $ \mathcal L_\ell^N$ but with a worse constant than the (presumably optimal) one that was derived for a restricted family of wave functions in ~\cite{RouYng-14}. This difference is due to the different methods employed. It remains an open problem to bridge the gap between~\cite{RouYng-14} and the present paper by proving the optimal bound without a priori assumptions on $F$.

\subsection{Main result}\label{sec:resul}

% We consider Laughlin's wave function 
% \begin{equation}\label{eq:intro Laughlin}
% \PsiLau (z_1,\ldots,z_N) = \cLau \prod_{1\leq i<j \leq N} (z_i-z_j) ^{\ell}  e ^{- \sum_{j=1} ^N |z_j| ^2 / 2} 
% \end{equation}
% where $\ell$ is a given positive integer. The constant $\cLau$ normalizes the state in $L^2(\R ^{2N})$. The set of fully correlated states built on $\PsiLau$ is given by
% \begin{equation}\label{eq:Ker}
% \Ker = \left\{  \Psi\in L^2(\R ^{2N})| \;  \Psi (z_1,\ldots,z_N) = \PsiLau (z_1,\ldots,z_N) F (z_1,\ldots,z_N), \; F \in \BargN \right\} 
% \end{equation}
% where $\BargN$ denotes the $N$-body bosonic Bargmann space 
% \begin{equation}\label{eq:BargN}
% \BargN :=  \left\{ F \mbox{ holomorphic and symmetric } | \:  F(z_1,\ldots,z_N) e^{- \sum_{j=1} ^N |z_j| ^2 /2  } \in L ^2 (\R ^{2N})\right\} 
% \end{equation}
% and the notation $\Ker$ is motivated by the fact that this space is precisely the kernel of some model two-body interaction Hamiltonians. 

As discussed at length in~\cite{RouYng-14}, our incompressibility estimates are formulated in terms of the minimum energy that one can achieve amongst states in $\mathcal L_\ell^N$ feeling an external potential which lives on the scale of the Laughlin state, i.e. over lengths of order $\sqrt{N}.$ This amounts to considering the energy functional 
\begin{equation}\label{eq:scale ener}
\E_N [\Psi] = (N-1) \intR V\left(\xbf\right) \rhoP \left(\sqrt{N-1} \: \xbf \right) 
\end{equation}
where $V$ is a given trapping potential, $\Psi \in \mathcal L_\ell^N$ and $\rhoP$ is the corresponding one-particle probability density~\eqref{eq:intro density}. Note the choice of normalization: we have set 
$$\intR (N-1)  \rhoP \left(\sqrt{N-1} \: \xbf \right) d\xbf = 1.$$
We shall be concerned with the large $N$ behavior of the ground state energy
\begin{equation}\label{eq:energy Nn}
E(N) := \inf \left\{ \E_N [\Psi_F], \; \Psi_F \in \mathcal L^N_\ell,\; \norm{\Psi}_{L^2} = 1 \right\}.
\end{equation} 
The following is our main result:

\begin{theorem}[\textbf{Unconditional incompressibility for $\mathcal{L}^N _\ell$}]\label{thm:main}\mbox{}\\
Let $V\in C ^{2} (\R^2)$ be increasing at infinity in the sense that
\begin{equation}\label{eq:increase V}
\min_{|x|\geq R} V(\xbf)\to \infty\quad\text{for}\quad R\to\infty.
\end{equation}
Define the corresponding {\it ``bathtub energy''} (cf. \cite[Theorem 1.14]{LieLos-01}) by
\begin{equation}\label{eq:bath tub}
E_V \left(\ell / 4 \right):= \inf\left\{ \intR V \rho \: \Big| \: \rho \in L^1 (\R ^2), \: 0 \leq \rho \leq \frac{4}{\pi \ell},\ \intR\rho=1 \right\}.
\end{equation}
Then 
\begin{equation}\label{eq:main incomp}
\liminf_{N\to \infty} E(N) \geq E_V\left( \ell / 4 \right).  
\end{equation}
\end{theorem}

This is a weak version of~\eqref{eq:incomp formel 2} in the sense that a universal lower bound to the energy $\E_N [\Psi_F]$ can be computed just as if~\eqref{eq:incomp formel 2} would hold, for a large class of one-body potentials. The notation $E_V\left( \ell / 4 \right)$ corresponds to the convention in~\cite[Theorem 2.1]{RouYng-14}. We \emph{conjecture} that the lower bound~\eqref{eq:main incomp} can be improved to  
\begin{equation}\label{eq:conj incomp}
\liminf_{N\to \infty} E(N) \geq E_V\left(\ell \right)  
\end{equation}
but this is out of reach with the method we use here. In~\cite[Theorem 2.1]{RouYng-14} we proved a conditional version of such a lower bound using a different method, i.e. we showed that \eqref{eq:conj incomp} holds under the a priori assumption (reasonable but not rigorously justified) that it suffices to consider correlation factors $F$ belonging to~\eqref{eq:var set 2 D}. Note that $E_V(\ell)$ is clearly an increasing function of $\ell$. For the homogeneous potential $V(x) = |x|^s$, it is easy to see that 
$$ E_V (\ell) = \frac{2}{s+2} \ell ^{s/2}.$$

Another natural conjecture is that the minimum energy $E(N)$ can always be achieved in the \emph{Laughlin phase} of functions of the form 
$$\Psi (z_1,\ldots,z_N) = c_{f_1} \PsiLau (z_1,\ldots,z_N) \prod_{j=1} ^N f_1 (z_j),\quad \norm{\Psi}_{L^2} = 1,$$
that is functions that do not contain any correlation factors beside the Laughlin state itself.  In view of~\cite[Corollary 2.3]{RouYng-14}, this would follow at least for some special radial potentials from the optimal bound~\eqref{eq:conj incomp}. This conjecture means informally that it is never favorable to leave the Laughlin phase by adding more correlations, whatever the one-body potential.

\subsection{Method of proof: The plasma analogy}\label{sec:method}

As in our previous works~\cite{RouSerYng-13b,RouSerYng-13,RouYng-14}, the starting point is Laughlin's plasma analogy, which is reminiscent of the log-gas analogy in random matrix theory~\cite{AndGuiZei-10,Dyson-62a,Forrester-10,Ginibre-65,Mehta-04}. The crucial observation~\cite{Laughlin-83,Laughlin-87} is that the absolute square of the wave-function~\eqref{eq:intro Laughlin} can be regarded as the Gibbs state of a 2D Coulomb gas (one-component plasma). Following our approach in~\cite{RouYng-14} we generalize this idea to map  any state of $\mathcal{L}_\ell ^N$ to a Gibbs state of a classical Hamiltonian.

Let us consider a state 
\begin{equation}\label{eq:start state}
\Psi_F = c_F \PsiLau (z_1,\ldots,z_N) F(z_1,\ldots,z_N) 
\end{equation}
where $\PsiLau$ is the Laughlin state \eqref{eq:intro Laughlin}, $F$ is holomorphic and symmetric and $c_F$ is a normalization factor. The idea is simply to write 
$$ |\Psi_F| ^2 = c_F ^2 \exp\left( - 2\log |\PsiLau| -2 \log |F| \right)$$
and interpret $2\log |\PsiLau| + 2\log |F|$ as a classical Hamilton function. One of the reasons why this is an effective procedure is that we can take advantage of the good scaling properties of $|\PsiLau| ^2$ to first change variables and obtain a classical Gibbs states \emph{with mean-field two-body interactions} and \emph{small effective temperature}. Specifically, we define
\begin{equation}\label{eq:Gibbs state}
\mu_N (Z) := (N-1)^{N} \left| \Psi_F \left(\sqrt{N-1} \: Z\right) \right| ^2 = \frac{1}{\ZN} \exp \left(-\frac{1}{T} H_N (Z) \right)
\end{equation}
where $\ZN$ ensures normalization of $\mu_N$ in $L^1(\R ^{2N})$,
$$T=\frac{1}{N},$$
and the classical Hamiltonian is of the form
\begin{equation}\label{eq:class Hamil}
H_N (Z) = \sum_{j=1} ^N |z_j| ^2 + \frac{2\ell}{N-1} \sum_{1 \leq i<j \leq N } w(z_i-z_j) + \frac{1}{N-1} W(Z).
\end{equation}
We have here written
\begin{equation}\label{eq:coul pot}
w(z) := - \log |z| 
\end{equation}
for the 2D Coulomb kernel and defined
\begin{equation}\label{eq:weird pot}
W(Z) := - 2 \log \left| F\left( \sqrt{N-1} \: Z  \right) \right|. 
\end{equation}
The $n$-th marginal of $\mu_N$ is defined as
$$ \mu_N ^{(n)}(z_1,\ldots,z_n):= \int_{\R ^{2{N-n}}} \mu_N (z_1,\ldots,z_N)dz_{n+1}\ldots dz_N.$$
We are ultimately interested in an incompressibility bound of the form (this is~\eqref{eq:incomp formel} after changing length and density units as in~\eqref{eq:Gibbs state})
\begin{equation}\label{eq:formal incomp}
\mu_N ^{(1)} \lesssim \frac{{\rm const.}}{\pi \ell} 
\end{equation}
in an appropriate weak sense, independently of the details of $W$. 

The function \eqref{eq:weird pot} can be rather intricate and represents in general a genuine $N$-body interaction term of the Hamiltonian \eqref{eq:class Hamil}. The only thing we know a priori about $W$ is that it is \emph{superharmonic in each of its variables}:
\begin{equation}\label{eq:subharm}
-\Delta_{z_j} W \geq 0 \quad \forall j=1\ldots N 
\end{equation}
which follows from the fact that $F$ is holomorphic. As explained in~\cite{RouYng-14}, under the additional assumption that $F$ belongs to $\VD$, it is possible to regard $W$ as a few-body interaction in a mean-field like scaling. Bounds of the form~\eqref{eq:formal incomp} where in~\cite{RouYng-14} shown to follow from a mean-field approximation for $\mu_N$: We first justified that an ansatz of the form
\begin{equation}\label{eq:MF ansatz}
\mu_N \approx \rho ^{\otimes N}, \quad \rho \in L ^1 (\R ^2) 
\end{equation}
is an effective ansatz for deriving  density bounds on $\mu_N ^{(1)}$.  It then followed from superharmonicity of $W$ that the appropriate $\rho$ must satisfy the desired incompressibility bound~\eqref{eq:formal incomp}, with the optimal constant 1. The difficult part in this approach was the justification of the ansatz~\eqref{eq:MF ansatz}. In the general case there does not seem to be any reason why~\eqref{eq:MF ansatz} should be a good approximation, and a new strategy is called for. 

There is, in fact, another way to bound the one-particle density in ground states of Coulomb systems without using any mean-field approximation. The argument is due to Lieb (unpublished), and variants thereof have been used recently in~\cite{PetSer-14,RotSer-13,RouSer-13}. It is based solely on some properties of the Coulomb kernel and general superharmonicity arguments, so that it can be adapted to Hamiltonians of the form~\eqref{eq:class Hamil}. Our strategy to prove Theorem~\ref{thm:main} is then the following:
\begin{itemize}
\item We adapt Lieb's argument to obtain a bound on the {\it minimal separation of points} in the ground state configurations
%{\red for a modification of the classical Hamiltonian~\eqref{eq:class Hamil}, including a small perturbation}
. This implies a local density upper bound at the level of the ground state,  but not yet the Gibbs state \eqref{eq:Gibbs state}.
\item We exploit the fact that~\eqref{eq:Gibbs state} is a Gibbs state for $H_N$ with small temperature $T \to 0$ in the limit $N\to \infty$. It  is thus reasonable to expect the density bound on the ground state to also apply to the Gibbs state.
\item More precisely, appropriate upper and lower bounds to the free-energy $- T \log \ZN$ confirm that it is close to the ground state energy of $H_N$ in the limit $T\to 0$. Applying these bounds to a suitably perturbed Hamiltonian gives the desired estimates on the density by a Feynman-Hellmann type argument.
\end{itemize}

This method has the merit of yielding weak density upper bounds, under the {\it sole} assumption that $F$ is holomorphic and symmetric. It can, however, not give the expected optimal bound~\eqref{eq:conj incomp} since it is based on a minimal distance estimate for ground state configurations: Obtaining the optimal constant with this method would be basically like computing the minimal distance between points in a Coulomb system, i.e., solving the crystallization problem for repulsive 2D Coulomb systems, a notoriously hard question (see~\cite{PetSer-14,SanSer-12b,RotSer-13,RouSer-13} for recent progress).

\medskip

The rest of the paper is devoted to the proof of Theorem~\ref{thm:main}. Section~\ref{sec:ground state} contains the analysis of the ground state of the classical Hamiltonian~\eqref{eq:class Hamil}. Adaptations to Gibbs states at small temperature are given in Section~\ref{sec:small T lim}. We come back to the original problem and conclude the proof in Section~\ref{sec:concl proof}.

\medskip

\noindent\textbf{Acknowledgments.} N.R. thanks the \textit{Erwin Schr\"odinger Institute}, Vienna, for its hospitality.  We received financial support from the ANR (project Mathostaq, ANR-13-JS01-0005-01) and the Austrian Science Fund (FWF) under project P~22929-N16.

\section{Estimates for ground state configurations}\label{sec:ground state}

In this section we discuss density bounds for the ground state configurations of $H_N$, i.e., for configurations of points $(z_1,\ldots,z_N)=Z_N\in \R^{2N}$ minimizing the function~\eqref{eq:class Hamil}. These density bounds will be derived from bounds on the minimal distance $\min_{i\neq j} \dist(z_i,z_j)$ as we explain first in Section~\ref{sec:incomp bound}. The lower bound on the minimal distance is then  proved in Section~\ref{sec:sep points}. 

\subsection{Density Upper bounds}\label{sec:incomp bound}
%\iffalse
Since we shall later use a Feynman-Hellmann argument, we need to discuss, besides $H_N$, the perturbed Hamiltonian
\begin{equation}\label{eq:pert Hamil}
H_N ^{\eps} (Z_N) = \sum_{j=1} ^N \left( |z_j| ^2 + \eps U(z_j)\right) + \frac{2\ell}{N-1} \sum_{1 \leq i<j \leq N } w(z_i-z_j) + \frac{1}{N-1} W(Z)
\end{equation}
where $U\in C^2(\R ^2)$ is a uniformly bounded function with uniformly bounded derivatives up to second order. To understand the main point  of the argumentation, it is sufficient to think of the case $\eps = 0$, however. We shall prove the following, which is reminiscent of previous results proved in~\cite{PetSer-14,RotSer-13,RouSer-13} in the case $W\equiv 0$.

\begin{proposition}[\textbf{Separation of points in a ground state configuration}]\label{pro:min dist}\mbox{}\\
Let $Z_N ^\eps  = (z_1 ^\eps,\ldots,z_N ^\eps)$ be a ground state configuration for $H_N ^{\eps}$. Then, for $\eps$ small enough,
\begin{equation}\label{eq:min dist}
\min_{i\neq j} \dist (z_i ^\eps,z_j ^\eps) \geq \sqrt{\frac{\ell}{N-1}}\left(1 - 4 \eps^{1/2}  \norm{\Delta U}_{L ^{\infty}}^{1/2} \right). 
\end{equation}
\end{proposition}
%\fi
To see that this indeed provides an incompressibility bound of the form~\eqref{eq:formal incomp} at the level of the ground state, we state the 

\begin{corollary}[\textbf{Weak incompressibility bound for the ground state}]\label{cor:dens bound}\mbox{}\\
Let $Z_N ^\eps = (z_1 ^\eps,\ldots,z_N^\eps)$ be a ground state configuration for $H_N ^{\eps}$ and 
\begin{equation}\label{eq:ground density}
\rho_0  ^\eps (z) = \frac{1}{N} \sum_{j=1} ^N \delta(z-z_j^\eps)
\end{equation}
be the corresponding $1$-particle density (empirical measure). There exists a $\rhot_0 ^\eps \in L ^{\infty} (\R ^2)$ satisfying 
$$\int_{\R ^2} \rhot_0 ^\eps = 1$$
and 
\begin{equation}\label{eq:incomp tilde}
\rhot_0  ^\eps \leq \frac{4}{\pi \ell} \left(1 +8 \eps^{1/2} \norm{\Delta U}_{L ^{\infty}}^{1/2} \right) \left( 1 - N ^{-1} \right)
\end{equation}
for $\eps$ small enough, such that for any $ f \in C^1 (\R^2)$ we have
\begin{equation}\label{eq:incomp estimate}
\left| \int_{\R ^2} \left( \rho_0 ^\eps - \rhot_0 ^\eps \right)  f \right| \leq C N ^{-1/2} \norm{\nabla  f}_{L ^{\infty}}
\end{equation}
in the limit $N\to\infty$, with a constant  $C\sim \ell^{1/2}$.
\end{corollary}

\begin{proof}
The obvious choice is to define 
\begin{equation}\label{eq:defi rhot}
\rhot_0 ^\eps (z) = \frac{1}{N} \sum_{j=1} ^N \frac{1}{\pi L ^2}  \one_{z\in B(z_j ^\eps,L)}
\end{equation}
where $\one_{z\in B(z_j ^\eps,L)}$ is the characteristic function of the ball $B(z_j ^\eps,L)$ with center $z_j$ and radius $L$, with
$$ L = \frac{1}{2}\ \sqrt{\frac{\ell}{N-1}}\left(1 - 4 \eps^{1/2}  \norm{\Delta U}_{L ^{\infty}}^{1/2} \right). $$
Then~\eqref{eq:incomp tilde} holds because~\eqref{eq:min dist} ensures that the balls $B(z_j ^\eps,L)$ do not overlap, and~\eqref{eq:incomp estimate} is obvious, with $ C \sim \ell ^{1/2}$ since the radius of the ball is $O(\sqrt{\ell/N})$.
\end{proof}

\subsection{Separation of points in the ground state}\label{sec:sep points}

To prove Proposition 2.1
%~\ref{pro:min dist} 
we use the fact that any point in a minimizing configuration must sit at the minimum of the potential generated by all the other points plus the external potential. To simplify the notation we drop the $\eps$ superscripts on the points $z_j  ^\eps$ in this proof.

\medskip

Define \beq L_0:=\sqrt{\frac \ell{N-1}},\eeq and, for $0<\delta<1$, 
\beq L_\delta:=L_0(1-\delta)\eeq
We prove Proposition 2.1 by showing that for $\delta=4\eps^{1/2} \Vert \Delta U\Vert_{L^\infty}^{1/2}$ and $\eps$ small enough so that $0<\delta<1$, the energy of any configuration $Z$ with two points, say $z_0$ and $z$, satisfying
\begin{equation}\label{eq:small dist}
\dist (z_0,z) < L_\delta 
\end{equation}
can be strictly lowered by moving $z$ to another point $z'$ farther from $z_0$, namely such that
$$\dist (z_0,z') = L_0>L_\delta.$$
To show this  we compute the energy difference between the two configurations\footnote{By symmetry of $H_N$ we may assume that $z_0$ and $z$ are the first two points in the labeling.}, $Z= \left( z_0, z, z_3,\ldots,z_N \right)$ and  $Z'=\left( z_0, z', z_3,\ldots,z_N \right)$: 
\begin{align*}
H_N (z_0,z,z_3,\ldots,z_N) &- H_N (z_0,z',z_3,\ldots,z_N) \\
&= |z| ^2 - |z'| ^2 + \frac{2 \ell}{N-1} w(z_0-z) - \frac{2 \ell}{N-1} w(z_0-z') 
\\&+ \frac{1}{N-1} W(Z) - \frac{1}{N-1} W (Z') + \frac{2\ell}{N-1} \sum_{j= 3} ^N \left( w (z_j-z) - w (z_j-z') \right)
\\&+ \eps \left( U(z) - U(z')\right).
\end{align*}
Considering the points $z_0,z_3,\ldots,z_N$ as fixed we write this as
\begin{equation}\label{eq:decomp}
 H_N (Z) - H_N (Z') = G(z) - G(z') + R(z) - R(z') + \eps \left(\tilde{U} (z) - \tilde{U} (z')\right) 
\end{equation}
where 
\begin{align*}
G(z) &= \frac{2 \ell}{N-1} w(z_0-z) + P (z)\\
P (z) &= \frac{2}{\pi}\log |\: . \:| \ast  \one_{z\in B (z_0,L_0)}\\
\tilde{U} (z) &= \frac{1}{2\pi} \log |\: . \:| \ast \left( (\Delta U) \one_{z\in B (z_0,L_0)} \right)
\end{align*} 
and 
 \beq\label{R} R (z) = \frac{2\ell}{N-1} \sum_{j= 3} ^N  w (z_j-z) + \frac{1}{N-1} W(z_0,z,z_3,\ldots,z_N) + |z| ^2 + \eps \left(U(z)-\tilde{U}(z)\right) - P (z).\eeq
We  first note that $z\mapsto R(z)$ is superharmonic on $B(z_0,L_0)$. Indeed,  the first two terms in \eqref{R} are by definition superharmonic everywhere (recall (1.17) and (1.20))
%(recall~\eqref{eq:coul pot} and~\eqref{eq:subharm})
 and for the combination of the last three we use that $P$ is precisely tuned in order to have 
$$ - \Delta \left(|z| ^2 + \eps \left(U(z)-\tilde{U}(z)\right) - P (z)\right) = 0$$
in $B(z_0,L_0)$. Hence 
$$ -\Delta R (z) \geq 0 \mbox{ for } z\in B(z_0,L_0).$$
As a consequence $R$ reaches its minimum over $B(z_0,L_0)$ at some point $z'$ of the boundary of the disc:
\begin{equation}\label{eq:vari R} 
R(z) -R(z') \geq 0 
\end{equation}
with $\dist (z_0,z') = L_0$.

Next we note that 
\begin{align*}
 \left| \nabla \tilde{U} \right| &\leq  \frac{2}{\pi} |\: . \:| ^{-1} \ast \left( |\Delta U| \one_{z\in B (z_0,L_0)}\right)\\
 &\leq 4 \norm{\Delta U}_{L ^{\infty}} \int_{B(0,L_0)} \frac{1}{|y|} dy \leq 4\, L_0 \norm{\Delta U}_{L ^{\infty}}.
\end{align*}
Since $|z-z'|\leq 2 L_0$ we thus obtain
\begin{equation}\label{eq:vari Ut}
\eps \left|\tilde{U} (z) - \tilde{U} (z')\right| \leq 8\,\eps  L_0^2\norm{\Delta U}_{L^{\infty}} .
\end{equation}

It remains to estimate $G(z) - G(z')$. Since $G$ is radial around $z_0$ and monotonously decreasing in the radial variable on $B(z_0,L_0)$ (see Eq.\eqref{gprime} below) we can just translate the whole system to set $z_0 = 0$ and compute $G(L_{\delta}) - G(L_0)$. But $G$ is the potential generated by a positive point charge sitting at $z_0=0$ and a negative charge smeared over the ball $B(0,L_0)$, so
$$ G = -\log |\: . \:| \ast \left( \frac{2\ell}{N-1}\delta_{0} - \frac{2}{\pi} \one_{B (0,L_0)}\right).$$ 
%the total charge is non-positive 
%$$ \int \left( \frac{2 \ell}{N-1}\delta_{0} - \frac{2}{\pi} \one_{B (0,L_^+\eps)}\right) = \frac{2 \ell}{N-1} - 2 L_\eps ^2 = 0$$
An explicit computation using Newton's theorem shows that for $r\leq L_0$ 
$$ G(r) = 2 (L_0) ^2 \log L_0 - \frac{2\ell}{N-1} \log r + r ^2 -  (L_0) ^2$$
and 
\beq\label{gprime} G'(r) =- \frac{2\ell}{N-1} \frac 1r + 2r <0\eeq
 if $r<L_0$. Moreover, $G(L_0)=0$ and, using that $-\log(1-\delta)\geq \delta-\half \delta^{2}$, we have
 $$G(L_\delta)\geq\half\delta^2L_0^2$$
 Thus, 
$$G(z)-G(z')+\eps(\tilde{U} (z) - \tilde{U} (z'))\geq G(z)-G(z')-\eps |\tilde{U} (z) - \tilde{U} (z')| >0$$
provided
\beq\label{eq:choice delta}
\half\,\delta^2  L_0^2-8\,\eps L_0^2 \norm{\Delta U}_{L^{\infty}} >0, \quad\hbox{i.e.}\quad \delta>4\,\eps^{1/2}\,\Vert \Delta U\Vert_{L\infty}^{1/2}.
\eeq
In view of~\eqref{eq:decomp} and~\eqref{eq:vari R} we thus have 
$$ H_N (Z) > H_N (Z').$$
Hence \eqref{eq:small dist} with $\delta$ as in~\eqref{eq:choice delta} cannot hold for a minimizing configuration, so~\eqref{eq:min dist} must hold in such a configuration and the proof of Proposition~\ref{pro:min dist} is complete. \hfill\qed

\section{Applications to Gibbs states at small temperature}\label{sec:thermal state}

In this section we deduce incompressibility bounds for the Gibbs state~\eqref{eq:Gibbs state} at $T = N ^{-1}$ from the zero-temperature results of the previous section. The main step is to  bound from below the energy in a regular potential $U$ with bounded derivatives, which should be thought of as a truncation of the physical potential $V$ appearing in Theorem~\ref{thm:main}, see Section~\ref{sec:concl proof} below.

\begin{proposition}[\textbf{Energy lower bounds in truncated potentials}]\label{pro:incomp pre}\mbox{}\\
Let $\mu_N$ be defined as in~\eqref{eq:Gibbs state}, let $U \in C^2 (\R ^2)$ be such that $U,\Delta U \in L ^{\infty} (\R ^2)$. For any $N$ large enough and $\eps>0$ small enough there exists a probability density  $\rho \in L ^1 (\R ^2)$ satisfying
\begin{equation}\label{eq:bound MF thm}
\rho \leq \frac{4}{\pi \ell} \left(1 +8 \eps^{1/2} \norm{\Delta U}_{L ^{\infty}}^{1/2} \right) \left( 1 - N ^{-1} \right)
\end{equation}
such that, for a constant $C<\infty$, 
\begin{equation}\label{eq:incomp two}
\int_{\R ^2} U \mu_N ^{(1)} \geq  \int_{\R ^2} U  \rho - \frac{C}{\eps N}\left( 1 + \log N\right) - C N ^{-1/2}\norm{\nabla U}_{L ^{\infty}}.
\end{equation} 
\end{proposition}

Note the necessary adjustment of $\eps$ in applications: the smaller it is, the closer the density bound~\eqref{eq:bound MF thm} is from the desired one, but the larger the error term in~\eqref{eq:incomp two}.  We prove this result in the next Section~\ref{sec:small T lim} and deduce from it our main Theorem~\ref{thm:main} in Section~\ref{sec:concl proof}.

\subsection{Small temperature limit}\label{sec:small T lim}

The proof consists of upper and lower bounds to the free energy
\begin{equation}\label{eq:free ener}
F_N ^\eps := \inf\left\{\F_N ^{\eps} [\mu], \quad \mu\in \PP(\R ^{2N}) \right\}
\end{equation}
where the free energy functional on the space $\PP(\R ^{2N})$ of probability measures on $\mathbb R^2$ is 
\begin{equation}\label{eq:free ener func}
\F_N ^{\eps} [\mu] := \int_{\R ^{2N}}  H_N ^{\eps} (Z_N) \mu (Z_N) dZ_N + N ^{-1}\int_{\R ^{2N}} \mu \log \mu
\end{equation}
with the perturbed Hamiltonian~\eqref{eq:pert Hamil}. It is well-known that the infimum is achieved by the Gibbs measure 
\begin{equation}\label{eq:Gibbs pert}
\mu_N ^{\eps} (Z_N) = \frac{1}{\ZN ^{\eps}} \exp \left( -N    H_N ^{\eps} (Z_N) \right) \in \PP(\R ^{2N}) 
\end{equation}
where the partition function satisfies $-N ^{-1} \log \ZN ^{\eps} = F_N ^{\eps}.$ The perturbation by the one-body potential $\eps U$ of the original Hamiltonian~\eqref{eq:class Hamil} will allow us to deduce density bounds from free-energy estimates, in the spirit of the Feynman-Hellmann principle. 

%The idea is that, since the effective temperature $N^{-1}\to 0$ in the $N\to \infty$ limit, the results we proved on ground state configurations in the previous section imply estimates on the Gibbs state. 

\subsubsection*{Free energy upper bound} By the variational principle we have 
\begin{align*}
F_N ^ \eps \leq \F_N ^{\eps} [\mu_N] = F_N ^{0} + \eps N \int_{\R ^2} U \mu_N ^{(1)}.
\end{align*}
To obtain an upper bound on $F_N ^0$ we use a trial state which is a regularization, over a length scale $\eta$ to be later optimized over, of a ground state configuration $Z_N ^0 = (z_1 ^0,\ldots,z_N ^0)$ for $H_N$:
\begin{equation}\label{eq:trial state}
\mu_N ^t (z_1,\ldots,z_N) := \left(\frac{1}{\pi \eta ^2}\right) ^N \one_{z_1 \in B(z_1 ^0,\eta)} \otimes \ldots \otimes \one_{z_N \in B(z_N ^0,\eta)}. 
\end{equation}
Note that this probability measure is not symmetric under particle exchange, but that is of no concern  because the symmetry of the Hamiltonian implies that the infimum in~\eqref{eq:free ener} is the same with or without symmetry constraint. The entropy of $\mu_N ^t$ is 
$$ S=-\int_{\R ^{2N}} \mu_N ^t \log \mu_N ^t = N \log (\pi \eta ^2),$$
and since the temperature is $T=N^{-1}$ the  contribution $-TS$ to the free energy is $-\log (\pi \eta ^2)$.
For the energetic part we note that 
$$ \int_{\R ^{2N}} H_N (Z_N) \mu_N ^t (Z_N) dZ_N \leq \sum_{j=1 }^N \frac{1}{\pi \eta ^2} \int_{B(z_j ^0,\eta)} |z| ^2 dz + \frac{2\ell}{N-1} \sum_{1 \leq i<j \leq N } w(z_i ^0-z_j ^0) + \frac{1}{N-1} W(Z_N ^0),$$
because superharmonicity in each variable of the function 
$$(z_1,\ldots,z_N) \mapsto \frac{2\ell}{N-1} \sum_{1 \leq i<j \leq N } w(z_i-z_j) + \frac{1}{N-1} W(Z_N)$$
implies that it must decrease upon taking an average over balls centered at the $z_i ^0$'s. For the one-body term we have, integrating in polar coordinates 
$$\frac{1}{\pi \eta ^2} \int_{B(z_j ^0,\eta)} |z| ^2 dz = \frac{1}{\pi \eta ^2} \int_{B(0,\eta)} |z+z_j ^0| ^2 dz = |z_j ^0| ^2 + \frac{1}{\pi \eta ^2} \int_{B(0,\eta)} |z| ^2 dz = |z_j ^0| ^2 + \frac{\eta ^2}{2}.$$
All in all we thus have, for $N$ large enough, 
\begin{align}\label{eq:free ener up bound}
F_N ^ \eps &\leq  H_N (Z_N ^0)+ \eps N \int_{\R ^2} U \mu_N ^{(1)} - \log (\pi) - \log \eta ^2 + N \frac{\eta ^2}{2} \nonumber\\
&\leq  \min_{\R ^{2N}} H_N + \eps N \int_{\R ^2} U \mu_N ^{(1)} + C \left( \log N + 1\right).
\end{align}
where we chose $\eta = (2N) ^{-1/2}$ to optimize the error.

\subsubsection*{Free energy lower bound} To bound $F_N^\eps$ from below we start with the entropy term: we first set 
$$ \nu = \pi^{-1} \exp(-|z| ^2).$$
Then
\begin{align*}
\int_{\R ^{2N}} \mu_N ^\eps \log \mu_N ^\eps &= \int_{\R ^{2N}} \mu_N ^{\eps} \log \frac{\mu_N ^{\eps}}{\nu ^{\otimes N}} + \int_{\R ^{2N}} \mu_N ^\eps \log \nu ^{\otimes N} 
\\&\geq  \int_{\R ^{2N}} \mu_N ^\eps \log \nu ^{\otimes N} 
\\&=  N\int_{\R ^2} (\mu_N ^{\eps}) ^{(1)} \log \nu = - N\int_{\R ^2} |z| ^2 (\mu_N ^{\eps}) ^{(1)} - N \log \pi 
\end{align*}
using positivity of the relative entropy of two probability measures. Adding the energy term we obtain 
\begin{align}%\label{eq:free ener low bound}
F_N ^\eps &=  \int_{\R ^{2N}} H_N ^\eps (Z_N) \mu_N ^\eps (Z_N) dZ_N + N ^{-1} \int_{\R ^{2N}} \mu_N ^\eps \log \mu_N ^\eps \nonumber\\
&\geq \int_{\R ^{2N}} \left(H_N ^\eps (Z_N) - N ^{-1} \sum_{j=1} ^N |z_j| ^2 \right) \mu_N ^\eps (Z_N)  dZ_N  - \log \pi\nonumber\\
&\geq \min_{Z_N \in \R ^{2N}} \left\{H_N ^\eps (Z_N) - N ^{-1} \sum_{j=1} ^N |z_j| ^2  \right\} - \log \pi\label{eq:curlybrace}\\
&\geq \min_{Z_N \in \R ^{2N}} H_N ^\eps (Z_N) - C = H_N ^{\eps} (Z_N ^{\eps}) - C \label{eq:curlybrace2}\\
&\geq H_N (Z_N^\eps) + N \eps \int_{\R ^2} U (z) \rho_0 ^{\eps} - C \nonumber\\
&\geq \min_{\R ^{2N}} H_N + N \eps \int_{\R ^2} U \rho_0 ^{\eps} - C\label{eq:free ener low bound}
\end{align}
where \eqref{eq:curlybrace2} follows from a simple estimate on the value of $\sum_{j= 1} ^N |z_j|^2$ at a minimum configuration of the curly brace in \eqref{eq:curlybrace} (see similar computations in~\cite[Section 3]{RouYng-14}), and 
$ \rho_0 ^{\eps}$ is the empirical measure of a minimizing configuration $Z_N ^{\eps}$ for $H_N ^{\eps}$ defined in 
\eqref{eq:ground density}.

\subsubsection*{Density bound}    Combining~\eqref{eq:free ener up bound} and~\eqref{eq:free ener low bound} we obtain 
$$ \int_{\R ^2} U \mu_N ^{(1)} \geq \int_{\R ^2} U \rho_0 ^{\eps} - \frac{C}{\eps N}\left( 1 + \log N\right)$$
and we may apply Corollary~\ref{cor:dens bound} to $\rho_0 ^\eps$ to obtain a probability density $\rho := \rhot_{0} ^\eps$ satisfying the requirements of Proposition~\ref{pro:incomp pre} and such that  \eqref{eq:incomp two} holds, i.e.,
$$ \int_{\R ^2} U \mu_N ^{(1)} \geq \int_{\R ^2} U \rho - \frac{C}{\eps N}\left( 1 + \log N\right) - C N ^{-1/2}\norm{\nabla U}_{L ^{\infty}}.$$
\hfill\qed

\subsection{Conclusion of the proof of Theorem~\ref{thm:main}}\label{sec:concl proof}

To conclude the proof of our main theorem we proceed as in~\cite[Section 4.1]{RouYng-14}. We pick some (sequence of) correlation factor(s) $F\in \BargN$, construct the corresponding  state(s) $\Psi\in\mathcal{L} ^N_\ell$, normalized in $L^2$, and Gibbs measure(s) $\mu_N$. We also pick a large constant $B$ (to be tuned later on) and define the truncated potential
\begin{equation}\label{eq:truncated potential}
V_B (\xbf):= \min\{V(\xbf),B\}. 
\end{equation}
Thanks to~\eqref{eq:increase V}, this potential is constant outside of some ball centered at the origin and satisfies the assumptions of Proposition~\ref{pro:incomp pre}. We may thus apply this result with $U=V_B$ and the correlation factor $F$ at hand. In view of~\eqref{eq:Gibbs state} we have 
\[
\mu_N ^{(1)} (z) = (N-1) \rhoP \left( \sqrt{N-1}\: z \right) 
\]
and the proposition implies that there exists a $\rho=\rho_F$ of unit $L ^1$ norm satisfying
\[
0\leq \rhoF \leq M_{\eps,N} := \frac{4}{\pi \ell} \left(1+8\eps^{1/2} \norm{\Delta V_B}_{L ^{\infty}}^{1/2} \right) \left( 1 - N ^{-1} \right)
\]
such that 
\begin{align}\label{eq:appli incomp two}
\E_N [\Psi_F] &\geq (N-1) \intR V_B (\xbf) \rhoPF \left( \sqrt{N-1}\: \xbf \right)d\xbf  \nonumber \\
&\geq  \int_{\R ^2} V_B\, \rho_F - \frac{C}{\eps N}\left( 1 + \log N\right) - C N ^{-1/2}\norm{\nabla V_B}_{L ^{\infty}}\nonumber\\
&\geq \inf\left\{ \intR V_B \,\rho,\: 0\leq \rho \leq M_{\eps,N}, \,\int_{\R ^2} \rho = 1 \right\} - \frac{C}{\eps N}\left( 1 + \log N\right) - C N ^{-1/2}\norm{\nabla V_B}_{L ^{\infty}}
\end{align}
Passing then to the limit $N\to \infty$ at fixed $\eps$ and $B$ we obtain
\[
\liminf_{N\to \infty}  \E_N [\Psi_F] \geq \inf_{\rho}\left\{ \intR V_B\, \rho,\: 0\leq \rho \leq \frac{4}{\pi \ell} \left(1  + 8\eps^{1/2} \norm{\Delta V_B}_{L ^{\infty}}^{1/2}  \right),\: \int_{\R ^2} \rho = 1 \right\}
\]
and, since the right-hand side no longer depends on $F$, we have 
\[
 \liminf_{N\to \infty} E(N) \geq \inf_{\rho} \left\{ \intR V_B\, \rho,\: 0\leq \rho \leq \frac{4}{\pi \ell} \left(1 + 8\eps^{1/2} \norm{\Delta V_B}_{L ^{\infty}}^{1/2} \right),\: \int_{\R ^2} \rho = 1 \right\}.
\]
We may then pass to the limit $\eps \to 0$ at fixed $B$: 
\[
\liminf_{N\to \infty}  E(N) \geq \inf_{\rho} \left\{ \intR V_B\, \rho,\: 0\leq \rho \leq \frac{4}{\pi \ell},\: \int_{\R ^2} \rho = 1 \right\}
\]
and finally to the limit $B\to \infty$, which yields
\[
\liminf_{N\to \infty}  E(N) \geq E_V \left(\ell /4 \right)
\]
as desired. We have used some continuity properties of the bath-tub energy~\eqref{eq:bath tub} as a function of the upper bound on the admissible trial states and the cut-off of the potential. These follow easily from the explicit formula for the minimum bath-tub energy see~\cite[Theorem 1.14]{LieLos-01}. Note in particular that for $B$ large enough, the bathtub energy in $V_B$ is in fact equal to the bathtub energy in $V$.

\hfill \qed

\bibliographystyle{siam}
%\bibliography{biblio_NR_Oct14}

\begin{thebibliography}{10}

\bibitem{AndGuiZei-10}
{\sc G.~Anderson, A.~Guionnet, and O.~Zeitouni}, {\em An introduction to random
  matrices}, Cambridge University Press, 2010.

\bibitem{Ciftja-06}
{\sc O.~Ciftj{\'{a}}}, {\em Monte {C}arlo study of {B}ose {L}aughlin wave
  function for filling factors $1/2$, $1/4$ and $1/6$}, Europhys. Lett., 74
  (2006), pp.~486--492.

\bibitem{Dyson-62a}
{\sc F.~J. Dyson}, {\em Statistical theory of the energy levels of a complex
  system. part {I}}, J. Math. Phys., 3 (1962), pp.~140--156.

\bibitem{Forrester-10}
{\sc P.~Forrester}, {\em Log-gases and random matrices}, London Mathematical
  Society Monographs Series, Princeton University Press, 2004.

\bibitem{Ginibre-65}
{\sc J.~Ginibre}, {\em Statistical ensembles of complex, quaternion, and real
  matrices}, J. Math. Phys., 6 (1965), pp.~440--449.

\bibitem{Girvin-04}
{\sc S.~Girvin}, {\em Introduction to the fractional quantum {H}all effect},
  S\'eminaire Poincar\'e, 2 (2004), pp.~54--74.

\bibitem{Laughlin-83}
{\sc R.~B. Laughlin}, {\em Anomalous quantum {H}all effect: An incompressible
  quantum fluid with fractionally charged excitations}, Phys. Rev. Lett., 50
  (1983), pp.~1395--1398.

\bibitem{Laughlin-87}
\leavevmode\vrule height 2pt depth -1.6pt width 23pt, {\em Elementary theory :
  the incompressible quantum fluid}, in The quantum {H}all effect, R.~E. Prange
  and S.~E. Girvin, eds., Springer, Heidelberg, 1987.

\bibitem{LewSei-09}
{\sc M.~Lewin and R.~Seiringer}, {\em Strongly correlated phases in rapidly
  rotating {B}ose gases}, J. Stat. Phys., 137 (2009), pp.~1040--1062.

\bibitem{LieLos-01}
{\sc E.~H. Lieb and M.~Loss}, {\em Analysis}, vol.~14 of Graduate Studies in
  Mathematics, American Mathematical Society, Providence, RI, second~ed., 2001.

\bibitem{Mehta-04}
{\sc M.~Mehta}, {\em Random matrices. Third edition}, Elsevier/Academic Press,
  2004.

\bibitem{PapBer-01}
{\sc T.~Papenbrock and G.~F. Bertsch}, {\em Rotational spectra of weakly
  interacting {B}ose-{E}instein condensates}, Phys. Rev. A, 63 (2001),
  p.~023616.

\bibitem{PetSer-14}
{\sc M.~Petrache and S.~Serfaty}, {\em Next order asymptotics and renormalized
  energy for {R}iesz interactions}, ArXiv e-prints,  (2014).

\bibitem{RotSer-13}
{\sc S.~Rota-Nodari and S.~Serfaty}, {\em Renormalized energy equidistribution
  and local charge balance in 2d {C}oulomb systems}, ArXiv e-prints,  (2013).

\bibitem{RouSer-13}
{\sc N.~Rougerie and S.~Serfaty}, {\em Higher dimensional {C}oulomb gases and
  renormalized energy functionals}, ArXiv e-prints,  (2013).

\bibitem{RouSerYng-13b}
{\sc N.~Rougerie, S.~Serfaty, and J.~Yngvason}, {\em Quantum {H}all phases and
  plasma analogy in rotating trapped bose gases}, J. Stat. Phys.,  (2013).

\bibitem{RouSerYng-13}
\leavevmode\vrule height 2pt depth -1.6pt width 23pt, {\em Quantum {H}all
  states of bosons in rotating anharmonic traps}, Phys. Rev. A, 87 (2013),
  p.~023618.

\bibitem{RouYng-14}
{\sc N.~Rougerie and J.~Yngvason}, {\em Incompressibility estimates for the
  {L}aughlin phase}, Comm. Math. Phys.,  (2014).

\bibitem{SanSer-12b}
{\sc E.~{Sandier} and S.~{Serfaty}}, {\em {2D Coulomb Gases and the
  Renormalized Energy}}, ArXiv e-prints,  (2012).

\bibitem{StoTsuGos-99}
{\sc H.~St\"{o}rmer, D.~Tsui, and A.~Gossard}, {\em The fractional quantum
  {H}all effect}, Rev. Mod. Phys., 71 (1999), pp.~S298--S305.

\bibitem{TruKiv-85}
{\sc S.~Trugman and S.~Kivelson}, {\em Exact results for the fractional quantum
  {H}all effect with general interactions}, Phys. Rev. B, 31 (1985), p.~5280.

\end{thebibliography}

\end{document}